\documentclass[reqno,12pt]{amsart}

\usepackage{amsmath}
\usepackage{latexsym}
\usepackage{amssymb}
\usepackage{hyperref}
\usepackage{graphicx}
\usepackage{epstopdf}
\usepackage{ifpdf}
\ifpdf
\fi

\newtheorem{theorem}{Theorem}[section]

\newtheorem{proposition}[theorem]{Proposition}

\newtheorem{assumption}[theorem]{Assumption}

\newcommand{\R}{{\mathbb R}}

\newcommand{\C}{{\mathbb C}}

\newcommand{\be}{\begin{equation}}
\newcommand{\ee}{\end{equation}}
\newcommand{\bea}{\begin{eqnarray}}
\newcommand{\eea}{\end{eqnarray}}
\newcommand{\ba}{\begin{array}}
\newcommand{\ea}{\end{array}}

\newcommand{\id}{\mathbb{I}}

\newcommand{\re}{\mathrm{Re}}

\newcommand{\lam}{\lambda}

\newcommand{\gam}{\gamma}

\newcommand{\Om}{\Omega}
\newcommand{\dta}{\delta}

\newcommand{\tha}{\theta}

\numberwithin{equation}{section}

\begin{document}
\title[RHP for the 3-wave equation on the half-line]{The unified method for the three-wave equation on the half-line}

\author[J.Xu]{Jian Xu}
\address{School of Mathematical Sciences\\
Fudan University\\
Shanghai 200433\\
People's  Republic of China}
\email{11110180024@fudan.edu.cn}

\author[E.Fan]{Engui Fan*}
\address{School of Mathematical Sciences, Institute of Mathematics and Key Laboratory of Mathematics for Nonlinear Science\\
Fudan University\\
Shanghai 200433\\
People's  Republic of China}
\email{correspondence author:faneg@fudan.edu.cn}


\date{\today}

\begin{abstract}
We present a Riemann-Hilbert problem formalism for the initial-boundary value problem of the three-wave equation:
 \[
 p_{ij,t}-\frac{b_i-b_j}{a_i-a_j}p_{ij,x}+\sum_k(\frac{b_k-b_j}{a_k-a_j}-\frac{b_i-b_k}{a_i-a_k})p_{ik}p_{kj}=0,\quad i,j,k=1,2,3,
 \]
on the half-line.

\end{abstract}

\maketitle

\section{\bf Introduction}

The 3-wave resonant interaction model described by the equations
\be\label{3wE}
\ba{ll}
 p_{ij,t}-\frac{b_i-b_j}{a_i-a_j}p_{ij,x}+\sum_k(\frac{b_k-b_j}{a_k-a_j}-\frac{b_i-b_k}{a_i-a_k})p_{ik}p_{kj}=0,\\[4pt]
 i,j,k=1, 2, 3; \ a_i\ne a_j,  b_i\ne b_j, \ {\rm for} \ i\ne j,
 \ea
\ee
 is  one of the important nonlinear models with numerous applications to physics
\cite{LYS}. The $3$- and $N$-wave interaction models describe a special class of wave-wave interactions
that are not sensitive on the physical nature of the waves and bear an universal
character. This explains why they find numerous applications in physics and attract the
attention of the scientific community over the last few decades \cite{ZM,ZM2,ZMNP,K,FT,M,KRB,BC} and
the references therin.
\par
The 3-wave equations can be solved through the inverse scattering method due to the fact that
equation (\ref{3wE})  admits a  Lax representation \cite{ZM2,ZMNP}.  But until the 1990s,  the inverse scattering method was pursued almost entirely for pure initial value problems.
In 1997, Fokas announced a new unified approach for the analysis of inital-boundary value problems for linear
and nonlinear integrable PDEs \cite{f1,f2, f3}. The Fokas method provides a generalization of the
inverse scattering formalism from initial value to IBV problems, and over the last fifteen years, this method has been used
to analyze boundary value problems for several of the most important integrable equations with $2\times 2$
Lax pairs, such as  KdV,  Schr\"{o}dinger,  sine-Gordon, and stationary
axisymmetric Einstein equations, see e.g. \cite{abmfs1,l2}. Just like the IST on the line, the unified method
yields an expression for the solution of an initial-boundary value problem  with  that  of a Riemann-Hilbert problem.
In particular, the asymptotic behavior of the solution can be analyzed in an effective way by employing  the
Riemann-Hilbert problem and  the steepest descent method introduced by Deift and Zhou \cite{dz}.
Recently, Lenells develop a methodology for analyzing initial-boundary value problems for integrable evolution equations
with Lax pairs involving $3\times 3$ matrices \cite{l3}. He also used this method to analyze the Degasperis-Procesi equation
in \cite{l4}.

Pelloni and Pinotsis also  studied  the boundary value problem of the $N-$wave equation by using the unified method \cite{pp}.
Recently, Gerdjikov and Grahovski considered  Cauchy problem of the 3-wave equation with with non-vanishing
initial values \cite{GG}.   In this paper we analyze the initial-boundary value problem of the three-wave equation (\ref{3wE}) on the
half-line.   Compared with these two papers,  there are two  differences
in  our paper.    The first difference
 is that we get the residue conditions of  matrix function $M$ in the Riemann-Hilbert problem  (see (\ref{resM}) in the next section 2).
  The  second difference  is that we  the jump matrix $J$   is explicitly constructed  (  see  the equations (\ref{Jmndef}) and (\ref{Sn}) in the next section 2).
 Of course, the initial-boundary value problem
for  the $3-$wave equation does not need to analysis the global relation, because the initial data and the boundary data are all known.

\par
 The  organization of the paper is as follows.   In the following section 2,  we perform the spectral analysis of the associated Lax pair.
And we formulate the main Riemann-Hilbert problem in section 3.

\section{\bf Spectral Analysis}

Our goal in this section is to define  analytic eigenfunctions of the Lax pair (\ref{3wlax})
which are suitable foe the formulation of a Riemann-Hilbert problem.

\subsection{Lax pair}

We first consider the three-wave equations (\ref{3wE}), with
$(x, t)\in\Om$, and  $\Om$ denoting the half-line domain
$$\Om= \{0 < x < +\infty, 0 < t < T \}$$
and $T > 0$ being a fixed final time. We denote the initial and boundary values by $p_{ij,0}(x)$ and $q_{ij,0}(t)$, respectively
\[
p_{ij,0}(x)=p_{ij}(x,0),\quad q_{ij,0}(t)=p_{ij}(0,t)
\]
with $p_{ij,0}(x)$ and $q_{ij,0}(t)$ are rapidly decaying.
Equation (\ref{3wE}) admits the following Lax representation \cite{LYS}
\be\label{3wlax}
\left\{
\ba{l}
\phi_x=M\phi,\\
\phi_t=N\phi,
\ea
\right.
\ee
where $M=i\lam A+P$ and $N=i\lam B+Q$, with
\be
\ba{ll}
A=\left(\ba{ccc}a_1&0&0\\0&a_2&0\\0&0&a_3\ea\right)&
P=\left(\ba{ccc}0&p_{12}&p_{13}\\p_{21}&0&p_{23}\\p_{31}&p_{32}&0\ea\right)\\
B=\left(\ba{ccc}b_1&0&0\\0&b_2&0\\0&0&b_3\ea\right)
&Q=\left(\ba{ccc}0&n_{12}p_{12}&n_{13}p_{13}\\n_{21}p_{21}&0&n_{23}p_{23}\\n_{31}p_{31}&n_{32}p_{32}&0\ea\right).
\ea
\ee
Obviously,  $trace(A)=trace(B)=0$.
We also assume that $a_1>a_2>a_3$ and $b_1<b_2<b_3$.
By using  transformation
$$\mu=\phi e^{-i\lam Ax-i\lam B t},$$
  we  change Lax pair (2.1) in the form
\be\label{3wstlax}
\left\{
\ba{l}
\mu_x-i\lam [A,\mu]=P\mu,\\
\mu_t-i\lam [B,\mu]=Q\mu,
\ea
\right.
\ee
which can be further   written in differential form as
\be\label{3wclosedf}
d(e^{-i\lam \hat Ax-i\lam\hat Bt}\mu)=W(x,t,\lam),
\ee
where
\be\label{Wdef}
W(x,t,\lam)=e^{-i\lam \hat Ax-i\lam\hat Bt}(Pdx+Qdt)\mu
\ee
and  $e^{\hat A}X=e^{A}Xe^{-A}$.

\subsection{ Spectral functions}
We define three eigenfunctions $\{\mu_j\}_1^3$ of (\ref{3wstlax}) by the Volterra integral equations
\be\label{mujdef}
\mu_j(x,t,\lam)=\id+\int_{\gam_j}e^{(i \lam \hat Ax+i \lam \hat B t)}W_j(x',t',\lam).\qquad j=1,2,3.
\ee
where $W_j$ is given by (\ref{Wdef}) with $\mu$ replaced with $\mu_j$,
and the contours $\{\gam_j\}_1^3$ are showed in Figure 1.
\begin{figure}[th]
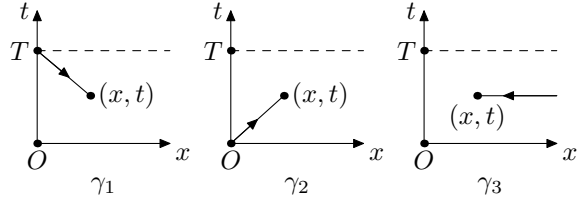

\centering
\includegraphics{3W--HL.1}
\includegraphics{3W--HL.2}
\includegraphics{3W--HL.3}
\caption{The three contours $\gam_1,\gam_2$ and $\gam_3$ in the $(x,t)-$domain.}
\end{figure}
And we have the following inequalities on the contours:
\be
\ba{ll}
\gam_1:&x-x'\ge 0,t-t'\le 0,\\
\gam_2:&x-x'\ge 0,t-t'\ge 0,\\
\gam_3:&x-x'\le 0.
\ea
\ee
So, these inequalities imply that the functions $\{\mu_j\}_1^3$ are bounded and analytic for
$\lam\in\C$ such that $\lam$ belongs to
\be\label{mujbodanydom}
\ba{ll}
\mu_1:&(D_2,\emptyset,D_1),\\
\mu_2:&\emptyset,\\
\mu_3:&(D_1,\emptyset,D_2),
\ea
\ee
where $\{D_n\}_1^2$ denote two open, pairwisely disjoint subsets of the complex $\lam$ plane showed in Figure 2.
\begin{figure}[th]
\centering
\includegraphics{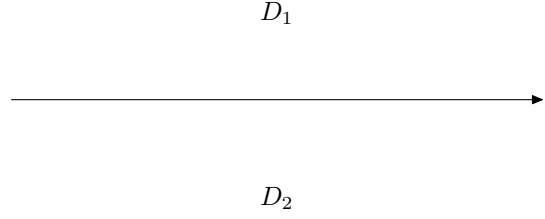}
\caption{The sets $D_n$, $n=1,2$, which decompose the complex $k-$plane.}
\end{figure}
\\
And the sets $\{D_n\}_1^2$ has the following properties:
\[
\ba{l}
D_1=\{\lam\in\C|\re{l_1}<\re{l_2}<\re{l_3},\re{z_1}>\re{z_2}>\re{z_3}\},\\
D_2=\{\lam\in\C|\re{l_1}>\re{l_2}>\re{l_3},\re{z_1}<\re{z_2}<\re{z_3}\},\\
\ea
\]
where $l_i(\lam)$ and $z_i(\lam)$ are the diagonal entries of matrices $i\lam A$ and $i\lam B$, respectively.

\subsection{Matrix valued FUNCTIONS  $M_n$'s}
For each $n=1,2$, define a solution $M_n(x,t,\lam)$ of (\ref{3wstlax}) by the following system of integral equations:
\be\label{Mndef}
(M_n)_{ij}(x,t,\lam)=\dta_{ij}+\int_{\gam_{ij}^n}(e^{(-i \lam \hat Ax-i \lam \hat Bt)}W_n(x',t',\lam))_{ij},
\quad \lam\in D_n,\quad i,j=1,2,3.
\ee
where $W_n$  an $M_n$  are given by (\ref{Wdef}), and the contours
$\gam_{ij}^n$, $n=1,2$, $i,j=1,2,3$ are defined by
\be\label{gamijndef}
\gam_{ij}^n=\left\{\ba{lclcl}\gam_1&if&\re l_i(\lam)<\re l_j(\lam)&and&\re z_i(\lam)\ge\re z_j(\lam),\\
\gam_2&if&\re l_i(\lam)<\re l_j(\lam)&and&\re z_i(\lam)<\re z_j(\lam),\\\gam_3&if&\re l_i(\lam)\ge\re l_j(\lam)&&.\\
\ea\right.\quad \mbox{for }\quad \lam\in D_n.
\ee
\par
The following proposition ascertains that the $M_n$'s defined in this way have the properties required
for the formulation of a Riemann-Hilbert problem.
\begin{proposition}
For each $n=1,2$, the function $M_n(x,t,\lam)$ is well-defined by equation (\ref{Mndef})
for $\lam\in \bar D_n$ and $(x,t)\in \Om$. For any fixed point $(x,t)$, $M_n$ is bounded and analytic as
a function of $\lam\in D_n$ away from a possible discrete set of singularities $\{\lam_j\}$ at which the
Fredholm determinant vanishes. Moreover, $M_n$ admits a bounded and continuous extension to $\bar D_n$ and
\be\label{Mnasy}
M_n(x,t,\lam)=\id+O(\frac{1}{\lam}),\qquad \lam\rightarrow \infty,\quad \lam\in D_n.
\ee
\end{proposition}
\begin{proof}
The bounedness and analyticity properties are established in appendix B in \cite{l3}. And substituting the expansion
\[
M=M_0+\frac{M^{(1)}}{\lam}+\frac{M^{(2)}}{\lam^2}+\cdots,\qquad \lam\rightarrow \infty.
\]
into the Lax pair (\ref{3wstlax}) and comparing the terms of the same order of $\lam$ yield the equation (\ref{Mnasy}).
\end{proof}

\subsection{The jump matrices}

We define spectral functions $S_n(\lam)$, $n=1,2$, and
\be\label{Sndef}
S_n(\lam)=M_n(0,0,\lam),\qquad \lam\in D_n,\quad n=1,2.
\ee
Let $M$ denote the sectionally analytic function on the complex $\lam-$plane which equals $M_n$ for $\lam\in D_n$.
Then $M_n$ satisfies the jump conditions
\be\label{Mjump}
M_1=M_2J,\qquad \lam\in \R,
\ee
where the jump matrices $J(x,t,\lam)$ are defined by
\be\label{Jmndef}
J=e^{(i \lam\hat Ax+i \lam\hat Bt)}(S_2^{-1}S_1).
\ee
According to the definition of the $\gam^n$, we find that
\be\label{gamndef}
\ba{ll}
\gam^1=\left(\ba{lll}\gam_3&\gam_1&\gam_1\\\gam_3&\gam_3&\gam_1\\\gam_3&\gam_3&\gam_3\ea\right)&
\gam^2=\left(\ba{lll}\gam_3&\gam_3&\gam_1\\\gam_3&\gam_3&\gam_3\\\gam_1&\gam_1&\gam_3\ea\right)\\
\ea
\ee

\subsection{The adjugated eigenfunctions}
We will also need the analyticity and boundedness properties of the minors of the matrices $\{\mu_j(x,t,\lam)\}_1^3$.
We recall that the adjugate matrix $X^A$ of a $3\times 3$ matrix $X$ is defined by
\[
X^A=\left(
\ba{ccc}
m_{11}(X)&-m_{12}(X)&m_{13}(X)\\
-m_{21}(X)&m_{22}(X)&-m_{23}(X)\\
m_{31}(X)&-m_{32}(X)&m_{33}(X)
\ea
\right),
\]
where $m_{ij}(X)$ denote the $(ij)$th minor of $X$.
\par
It follows from (\ref{3wstlax}) that the adjugated eigenfunction $\mu^A$ satisfies the Lax pair
\be\label{muadgLaxe}
\left\{
\ba{l}
\mu_x^A+[i\lam A,\mu^A]=-P^T\mu^A,\\
\mu_t^A+[i\lam B,\mu^A]=-Q^T\mu^A.
\ea
\right.
\ee
where $V^T$ denote the transform of a matrix $V$.
Thus, the eigenfunctions $\{\mu_j^A\}_1^3$ are solutions of the integral equations
\be\label{muadgdef}
\mu_j^A(x,t,\lam)=\id-\int_{\gam_j}e^{-i\lam\hat A(x-x')-i\lam B(t-t')}(P^Tdx+Q^T)\mu^A,\quad j=1,2,3.
\ee
Then we can get the following analyticity and boundedness properties:
\be\label{mujadgbodanydom}
\ba{ll}
\mu_1^A:&(D_1,\emptyset,D_2),\\
\mu_2^A:&\emptyset,\\
\mu_3^A:&(D_2,\emptyset,D_1).
\ea
\ee

\subsection{The  computation of jump matrices }

Let us define the $3\times 3-$matrix value spectral functions $s(\lam)$ and $S(\lam)$ by
\begin{subequations}\label{sSdef}
\be\label{mu3mu2s}
\mu_3(x,t,\lam)=\mu_2(x,t,\lam)e^{(i\lam \hat Ax+i\lam\hat Bt)}s(\lam),
\ee
\be\label{mu1mu2S}
\mu_1(x,t,\lam)=\mu_2(x,t,\lam)e^{(i\lam \hat Ax+i\lam\hat Bt)}S(\lam).
\ee
\end{subequations}
Thus,
\be\label{sSmu3mu1}
s(\lam)=\mu_3(0,0,\lam),\qquad S(\lam)=\mu_1(0,0,\lam).
\ee
And we deduce from the properties of $\mu_j$ and $\mu_j^A$ that $s(\lam)$ and $S(\lam)$ have the following
boundedness properties:
\[
\ba{ll}
s(\lam):&(D_1,\emptyset,D_2),\\
S(\lam):&(D_2,\emptyset,D_1),\\
s^A(\lam):&(D_2,\emptyset,D_1),\\
S^A(\lam):&(D_1,\emptyset,D_2).
\ea
\]
Moreover,
\be\label{MnSnrel}
M_n(x,t,\lam)=\mu_2(x,t,\lam)e^{(i\lam \hat Ax+i\lam\hat Bt)}S_n(\lam),\quad \lam\in D_n.
\ee
\begin{proposition}
The $S_n$ can be expressed in terms of the entries of $s(\lam)$ and $S(\lam)$ as follows:
\be\label{Sn}
\ba{l}
S_1=\left(\ba{ccc}s_{11}&\frac{m_{33}(s)M_{21}(S)-m_{23}(s)M_{31}(S)}{(s^TS^A)_{11}}&\frac{S_{13}}{(S^Ts^A)_{33}}\\
s_{21}&\frac{m_{33}(s)M_{11}(S)-m_{13}(s)M_{31}(S)}{(s^TS^A)_{11}}&\frac{S_{23}}{(S^Ts^A)_{33}}\\
s_{31}&\frac{m_{23}(s)M_{11}(S)-m_{13}(s)M_{21}(S)}{(s^TS^A)_{11}}&\frac{S_{33}}{(S^Ts^A)_{33}}\ea\right),\\
S_2=\left(\ba{ccc}\frac{S_{11}}{(S^Ts^A)_{11}}&\frac{m_{21}(s)M_{33}(S)-m_{31}(s)M_{23}(S)}{(s^TS^A)_{33}}&s_{13}\\
\frac{S_{21}}{(S^Ts^A)_{11}}&\frac{m_{11}(s)M_{33}(S)-m_{31}(s)M_{13}(S)}{(s^TS^A)_{33}}&s_{23}\\
\frac{S_{31}}{(S^Ts^A)_{11}}&\frac{m_{11}(s)M_{23}(S)-m_{21}(s)M_{13}(S)}{(s^TS^A)_{33}}&s_{33}\ea\right),
\ea
\ee
where $m_{ij}$ and $M_{ij}$ denote that the $(i,j)-$th minor of $s$ and $S$, respectively.
\end{proposition}
\begin{proof}
Let $\gam_3^{X_0}$ denote the contour $(X_0,0)\rightarrow (x,t)$ in the $(x,t)-$plane, here $X_0>0$ is a constant.
We introduce $\mu_3(x,t,k;X_0)$ as the solution of (\ref{mujdef}) with $j=3$ and with the contour $\gam_3$ replaced
by $\gam_3^{X_0}$. Similarly, we define $M_n(x,t,\lam;X_0)$ as the solution of (\ref{Mndef}) with $\gam_3$ replaced
by $\gam_3^{X_0}$. We will first derive expression for $S_n(\lam;X_0)=M_n(0,0,\lam;X_0)$ in terms of $S(\lam)$ and
$s(\lam;X_0)=\mu_3(0,0,\lam;X_0)$. Then (\ref{Sn}) will follow by taking the limit $X_0\rightarrow \infty$.
\par
First, We have the following relations:
\be\label{MnRnSnTn}
\left\{
\ba{l}
M_n(x,t,\lam;X_0)=\mu_1(x,t,\lam)e^{(i\lam \hat Ax+i\lam\hat Bt)}R_n(\lam;X_0),\\
M_n(x,t,\lam;X_0)=\mu_2(x,t,\lam)e^{(i\lam \hat Ax+i\lam\hat Bt)}S_n(\lam;X_0),\\
M_n(x,t,\lam;X_0)=\mu_3(x,t,\lam)e^{(i\lam \hat Ax+i\lam\hat Bt)}T_n(\lam;X_0).
\ea
\right.
\ee
Then we get $R_n(\lam;X_0)$ and $T_n(\lam;X_0)$ are defined as follows:
\begin{subequations}\label{RnTnX0}
\be\label{RnX0}
R_n(\lam;X_0)=e^{-i\lam\hat BT}M_n(0,T,\lam;X_0),
\ee
\be\label{TnX0}
T_n(\lam;X_0)=e^{-i\lam\hat AX_0}M_n(X_0,0,\lam;X_0).
\ee
\end{subequations}
The relations (\ref{MnRnSnTn}) imply that
\be\label{sSRnSnTn}
s(\lam;X_0)=S_n(\lam;X_0)T^{-1}_n(\lam;X_0),\qquad S(\lam)=S_n(\lam;X_0)R^{-1}_n(\lam;X_0).
\ee
These equations constitute a matrix factorization problem which, given $\{s(\lam),S(\lam)\}$ can be solved for the $\{R_n,S_n,T_n\}$.
Indeed, the integral equations (\ref{Mndef}) together with the definitions of $\{R_n,S_n,T_n\}$ imply that
\be
\left\{
\ba{lll}
(R_n(\lam;X_0))_{ij}=0&if&\gam_{ij}^n=\gam_1,\\
(S_n(\lam;X_0))_{ij}=0&if&\gam_{ij}^n=\gam_2,\\
(T_n(\lam;X_0))_{ij}=0&if&\gam_{ij}^n=\gam_3.
\ea
\right.
\ee
It follows that (\ref{sSRnSnTn}) are 18 scalar equations for 18 unknowns. By computing the explicit solution of this
algebraic system, we find that $\{S_n(\lam;X_0)\}_1^2$ are given by the equation obtained from (\ref{Sn}) by replacing
$\{S_n(\lam),s(\lam)\}$ with $\{S_n(\lam;X_0),s(\lam;X_0)\}$. taking $X_0\rightarrow \infty$ in this equation,
we arrive at (\ref{Sn}).
\end{proof}

\subsection{The residue conditions}
Since $\mu_2$ is an entire function, it follows from (\ref{MnSnrel}) that M can only have singularities at the points
where the $S_n's$ have singularities. We infer from the explicit formulas (\ref{Sn}) that the possible singularities
of $M$ are as follows:
\begin{itemize}
\item $[M]_1$ could have poles in $D_2$ at the zeros of $(S^Ts^A)_{11}(\lam)$;
\item $[M]_2$ could have poles in $D_1$ at the zeros of $(s^TS^A)_{11}(\lam)$;
\item $[M]_2$ could have poles in $D_2$ at the zeros of $(s^TS^A)_{33}(\lam)$;
\item $[M]_3$ could have poles in $D_1$ at the zeros of $(S^Ts^A)_{33}(\lam)$.
\end{itemize}
We denote the above possible zeros by $\{\lam_j\}_1^N$ and assume they satisfy the following assumption.
\begin{assumption}
We assume that
\begin{itemize}
\item $(S^Ts^A)_{11}(\lam)$ has $n_0$ possible simple zeros in $D_2$ denoted by $\{\lam_j\}_1^{n_0}$;
\item $(s^TS^A)_{11}(\lam)$ has $n_1-n_0$ possible simple zeros in $D_1$ denoted by $\{\lam_j\}_{n_0+1}^{n_1}$;
\item $(s^TS^A)_{33}(\lam)$ has $n_2-n_1$ possible simple zeros in $D_2$ denoted by $\{\lam_j\}_{n_1+1}^{n_2}$;
\item $(S^Ts^A)_{33}(\lam)$ has $n_3-n_2$ possible simple zeros in $D_1$ denoted by $\{\lam_j\}_{n_2+1}^{N}$;
\end{itemize}
and that none of these zeros coincide. Moreover, we assume that none of these functions have zeros on the boundaries
of the $D_n$'s.
\end{assumption}
We determine the residue conditions at these zeros in the following:
\begin{proposition}\label{propos}
Let $\{M_n\}_1^2$ be the eigenfunctions defined by (\ref{Mndef}) and assume that the set $\{\lam_j\}_1^N$ of singularities
are as the above assumption. Then the following residue conditions hold:
\begin{subequations}\label{resM}
\begin{align}
&\ba{l}{Res}_{\lam=\lam_j}[M]_1=\frac{1}{\dot{ (S^Ts^A)_{11}(\lam_j)}}\frac{(S_{11}s_{23}-S_{21}s_{13})(\lam_j)}{m_{31}(\lam_j)}e^{\tha_{21}(\lam_j)}[M(\lam_j)]_2,\\
\quad 1\le j\le n_0,\lam_j\in D_2\ea,\label{M21D2res}\\
&\ba{l}{Res}_{\lam=\lam_j}[M]_2=-\frac{1}{\dot{ (s^TS^A)_{11}(\lam_j)}}\frac{M_{21}(S^Ts^A)_{33}(\lam_j)}{ (S_{13}(\lam_j)s_{31}(\lam_j)-S_{33}(\lam_j)s_{11}(\lam_j))}e^{\tha_{12}(\lam_j)}[M(\lam_j)]_1,\\
\quad n_0< j\le n_1,\lam_j\in D_1\ea,\label{M12D1res}\\
&\ba{l}{Res}_{\lam=\lam_j}[M]_2=-\frac{1}{\dot{ (s^TS^A)_{33}(\lam_j)}}\frac{M_{23}(S^Ts^A)_{11}(\lam_j)}{ (s_{13}(\lam_j)S_{31}(\lam_j)-s_{33}(\lam_j)S_{11}(\lam_j))}e^{\tha_{32}(\lam_j)}[M(\lam_j)]_3,\\
\quad n_1< j\le n_2,\lam_j\in D_2\ea,\label{M22D2res}\\
&\ba{l}{Res}_{\lam=\lam_j}[M]_3=\frac{1}{\dot{ (S^Ts^A)_{33}(\lam_j)}}\frac{(S_{13}s_{21}-S_{23}s_{11})(\lam_j)}{m_{33}(\lam_j)}e^{\tha_{23}(\lam_j)}[M(\lam_j)]_2,\\
\quad n_2< j\le N,\lam_j\in D_1\ea,\label{M13D1res}
\end{align}
\end{subequations}
where $\dot f=\frac{df}{d\lam}$, and $\tha_{ij}$ is defined by
\be\label{thaijdef}
\tha_{ij}(x,t,\lam)=(l_i-l_j)x+(z_i-z_j)t,\quad i,j=1,2,3.
\ee
\end{proposition}

\begin{proof}
We will prove (\ref{M21D2res}), (\ref{M12D1res}), the other conditions follow by similar arguments.
Equation (\ref{MnSnrel}) implies the relation
\begin{subequations}
\be\label{M1S1}
M_1=\mu_2e^{(i\lam \hat Ax+i\lam\hat Bt)}S_1,
\ee
\be\label{M2S2}
M_2=\mu_2e^{(i\lam \hat Ax+i\lam\hat Bt)}S_2.
\ee
\end{subequations}
In view of the expressions for $S_1$ and $S_2$ given in (\ref{Sn}), the three columns of (\ref{M1S1}) read:
\begin{subequations}
\begin{align}
&[M_1]_1=[\mu_2]_1s_{11}(\lam)+[\mu_2]_2e^{\tha_{21}}s_{21}(\lam)+[\mu_2]_3 e^{\tha_{31}}s_{31}(\lam),\label{M11}\\
&\ba{ll}[M_1]_2=&[\mu_2]_1e^{\tha_{12}}\frac{m_{33}M_{21}-m_{23}M_{31}}{(s^TS^A)_{11}}(\lam)+
[\mu_2]_2\frac{m_{33}M_{11}-m_{13}M_{31}}{(s^TS^A)_{11}}(\lam)\\&+
[\mu_2]_3e^{\tha_{32}}\frac{m_{23}M_{11}-m_{13}M_{21}}{(s^TS^A)_{11}}(\lam)\ea,\label{M12}\\
&[M_1]_3=[\mu_2]_1e^{\tha_{13}}\frac{S_{13}}{(S^Ts^A)_{33}}(\lam)+
[\mu_2]_2e^{\tha_{23}}\frac{S_{23}}{(S^Ts^A)_{33}}(\lam)
+[\mu_2]_3\frac{S_{33}}{(S^Ts^A)_{33}}(\lam).\label{M13}
\end{align}
\end{subequations}
while the three columns of (\ref{M2S2}) read:
\begin{subequations}
\begin{align}
&[M_2]_1=[\mu_2]_1\frac{S_{11}}{(S^Ts^A)_{11}}(\lam)+
[\mu_2]_2e^{\tha_{21}}\frac{S_{21}}{(S^Ts^A)_{11}}(\lam)
+[\mu_2]_3e^{\tha_{31}}\frac{S_{31}}{(S^Ts^A)_{11}}(\lam),\label{M21}\\
&\ba{ll}[M_2]_2=&[\mu_2]_1e^{\tha_{12}}\frac{m_{21}M_{33}-m_{31}M_{23}}{(s^TS^A)_{33}}(\lam)+
[\mu_2]_2\frac{m_{11}M_{33}-m_{31}M_{13}}{(s^TS^A)_{33}}(\lam)\\&+
[\mu_2]_3e^{\tha_{32}}\frac{m_{11}M_{23}-m_{21}M_{13}}{(s^TS^A)_{33}}(\lam)\ea,\label{M22}\\
&[M_2]_3=[\mu_2]_1s_{13}e^{\tha_{13}}+[\mu_2]_2s_{23}e^{\tha_{23}}+[\mu_2]_3s_{33}.\label{M23}
\end{align}
\end{subequations}
We first suppose that $\lam_j\in D_2$ is a simple zero of $(S^Ts^A)_{11}(\lam)$. Solving (\ref{M22})
and (\ref{M23}) for $[\mu_2]_1$ and $[\mu_2]_2$ and substituting the result in to (\ref{M21}), we find
\[
\ba{rl}[M_1]_1=&\frac{S_{11}s_{23}-S_{21}s_{13}}{(S^Ts^A)_{11}m_{31}}e^{\tha_{21}}[M_2]_2+
\frac{M_{33}}{m_{31}}e^{\tha_{31}}[M_2]_3\\
&+\frac{1}{m_{31}}e^{\tha_{31}}[\mu_2]_3\ea.
\]
Taking the residue of this equation at $\lam_j$, we find the condition (\ref{M21D2res}) in the case
when $\lam_j\in D_2$. Similarly, we can get the equation (\ref{M13D1res}).
\par
Then let us consider that $\lam_j\in D_1$ is a simple zero of $(s^TS^A)_{11}(\lam)$. Solving (\ref{M11})
and (\ref{M13}) for $[\mu_2]_1$ and $[\mu_2]_3$ and substituting the result in to (\ref{M12}), we find
\[
\ba{rl}[M_1]_2=&-\frac{M_{21}(S^Ts^A)_{33}}{(s^TS^A)_{11}(S_{13}s_{31}-S_{33}s_{11})}e^{\tha_{12}}[M_1]_1-
\frac{(S^Ts^A)_{33}}{S_{13}s_{31}-S_{33}s_{11}}[\mu_2]_2\\
&-\frac{m_{23}(S^Ts^A)_{33}}{S_{13}s_{31}-S_{33}s_{11}}e^{\tha_{32}}[M_1]_3\ea.
\]
Taking the residue of this equation at $\lam_j$, we find the condition (\ref{M12D1res}) in the case
when $\lam_j\in D_1$. Similarly, we can get the equation (\ref{M22D2res}).
\end{proof}

\section{\bf The Riemann-Hilbert problem}

The sectionally analytic function $M(x,t,\lam)$ defined in section 2 satisfies a Riemann-Hilbert problem
which can be formulated in terms of the initial and boundary values of $p_{ij}(x,t)$. By solving this
Riemann-Hilbert problem, the solution of (\ref{3wE}) can be recovered for all values of $x,t$.
\begin{theorem}
Suppose that $p_{ij}(x,t)$ are a solution of (\ref{3wE}) in the half-line domain $\Om$ with sufficient
smoothness and decays as $x\rightarrow \infty$. Then $p_{ij}(x,t)$ can be reconstructed from the initial
value $\{p_{ij,0}(x)\}_{i,j=1}^3$ and boundary values $\{q_{ij,0}(t)\}_{i,j=1}^3$ defined as follows,
\be\label{inibouvalu}
p_{ij,0}(x)=p_{ij}(x,0),\quad q_{ij,0}(t)=p_{ij}(0,t).
\ee
\par
Use the initial and boundary data to define the jump matrices $J(x,t,\lam)$ as well as the spectral $s(\lam)$
and $S(\lam)$ by equation (\ref{sSdef}). Assume that the possible zeros $\{\lam_j\}_1^N$ of the functions $(S^Ts^A)_{33}(\lam),(s^TS^A)_{11}(\lam)$, $(s^TS^A)_{33}(\lam),(S^Ts^A)_{33}(\lam)$  are as in assumption 2.3.
\par
Then the solution $\{p_{ij}(x,t)\}_{i,j=1}^3$ is given by
\be\label{usolRHP}
p_{ij}(x,t)=-i(a_i-a_j)\lim_{\lam\rightarrow \infty}(\lam M(x,t,\lam))_{ij}.
\ee
where $M(x,t,\lam)$ satisfies the following $3\times 3$ matrix Riemann-Hilbert problem:
\begin{itemize}
\item $M$ is sectionally meromorphic on the complex $\lam-$plane with jumps across the contour $\R$, see Figure 2.
\item Across the contour $\R$, $M$ satisfies the jump condition
      \be\label{MRHP}
      M_1(x,t,\lam)=M_2(x,t,\lam)J(x,t,\lam),\quad \lam\in \R.
      \ee
      where the jump $J$ is defined by the equation (\ref{Jmndef}).
\item $M(x,t,\lam)=\id+O(\frac{1}{\lam}),\qquad \lam\rightarrow \infty$.
\item The residue condition of $M$ is showed in Proposition \ref{propos}.
\end{itemize}
\end{theorem}
\begin{proof}
It only remains to prove (\ref{usolRHP}) and this equation follows from the large $\lam$ asymptotics of the eigenfunctions.
\par
We write the large $\lam$ asymptotics of $M$ as follows:
\be\label{Mlargelamda}
M(x,t,\lam)=M_0(x,t)+\frac{M_1(x,t)}{\lam}+\cdots.\qquad \lam\rightarrow+\infty.
\ee
And insert this equation into the equation (\ref{3wstlax}) and compare the coeffients of the same order $\lam$,
for the $O(\lam)$, we have $M_0$ is a diagonal matrix ; for the $O(1)$, we get $M_0=\id$ by comparing the diagonal elements,
and we can have the following equation by comparing the other elements
\be\label{pij}
-i[A,M_1]=P,
\ee
this equation reads the required result of $p_{ij}(x,t)$
\be
p_{ij}(x,t)=-i(a_i-a_j)M_{1,ij}(x,t).
\ee
\end{proof}

{\bf Acknowledgements}
The work of  Xu was partially supported by Excellent Doctor Research Funding Project of  Fudan University.
The work described in this paper
was supported by grants from the National Science
Foundation of China (Project No.10971031;11271079), Doctoral Programs Foundation of
the Ministry of Education of China, and the Shanghai Shuguang Tracking Project (project 08GG01).

\end{document}